\documentclass[10pt,conference,a4paper]{IEEEtran}
\usepackage{epsfig}
\usepackage{cite}
\usepackage{amsfonts}
\usepackage{graphics}
\usepackage{amsmath}
\usepackage{amssymb}
\usepackage{latexsym}
\usepackage[]{fontenc}
\usepackage{psfrag}
\usepackage{dsfont}
\usepackage{booktabs}
\usepackage{bm}
\usepackage{algorithm}
\usepackage{algpseudocode}
\usepackage{algpascal}
\newtheorem{lemma}{Lemma}
\newtheorem{theorem}{Theorem}
\newtheorem{proposition}{Proposition}

\newtheorem{proof}{Proof}

\def\bg{{\boldsymbol{g}}}

\def\bn{{\boldsymbol{n}}}

\def\bp{{\boldsymbol{p}}}

\def\bx{{\boldsymbol{x}}}
\def\by{{\boldsymbol{y}}}

\def\bG{{\boldsymbol{G}}}

\def\bI{{\boldsymbol{I}}}

\IEEEoverridecommandlockouts
\begin{document}
\title{Uplink Pilot and Data Power Control for Single Cell Massive MIMO Systems with MRC}

\author{\IEEEauthorblockN{Hei Victor Cheng, Emil Bj\"ornson, and Erik G. Larsson}

\IEEEauthorblockA{Department of Electrical Engineering (ISY), Link\"oping University, Sweden\\
Email: \{hei.cheng, emil.bjornson, erik.g.larsson\}@liu.se}

\thanks{This work was supported by the Strategic Research Center, the Link\"oping University Center for Industrial Information Technology (CENIIT), and the EU FP7 Massive MIMO for Efficient Transmission (MAMMOET) project.}
}

\maketitle

\begin{abstract}
 This paper considers the jointly optimal pilot and data power allocation in single cell uplink massive MIMO systems. A closed form solution for the optimal length of the training interval is derived. Using the spectral efficiency (SE) as performance metric and setting a total energy budget per coherence interval the power control is formulated as optimization problems for two different objective functions: the minimum SE among the users and the sum SE. The optimal power control policy is found for the case of maximizing the minimum SE by converting it to a geometric program (GP). Since maximizing the sum SE is an $\mathrm{NP}$-hard problem, an efficient algorithm is developed for finding KKT (local maximum) points. Simulation results show the advantage of optimizing the power control over both pilot and data power, as compared to heuristic power control policies.
\end{abstract}
\section{Introduction}
Massive multiple-input-multiple-output (MIMO) have recently attracted a lot of attention\cite{LTEM2013,TM2010}. The idea of massive MIMO is to use a large amount of antennas at the base station (BS) to serve multiple users at the same time and frequency resource block. The ability to increase both spectral efficiency (SE) and energy efficiency makes it one of the key candidates for the 5G cellular networks. The analysis of the performance of massive MIMO is of vast importance and has been done in \cite{NLM2013} for uplink single cell systems and in \cite{HBD2013} for multi-cell systems. However the analysis is done with the assumption of equal power allocation among the users. Only a few papers has considered power control, however there has not been much optimization of the powers. In order to harvest all the benefits brought by the massive antenna arrays, power control among the users is necessary. This can be done by varying the power of different users to either increase the total system performance or provide services with certain fairness.

Power control in wireless cellular network has been an important problem for decades, dating back to point to point (P2P) wireless systems. Lots of efforts have been put into developing efficient algorithms for maximizing the system performance with different objectives. Due to the interference from other users the power control is usually hard to solve optimally, in particular NP-hardness was proven in \cite{LZ2008} for the objective of maximizing the sum performance. For practical use a reasonable approach is to develop suboptimal algorithms with affordable complexity while achieving an acceptable performance, as done for example in \cite{CTPNJ2007}.

Compared to power control in P2P systems, power control in massive MIMO networks is a relatively new topic. Accurate channel estimates are needed at the BS for carrying out coherent linear processing, e.g. uplink detection and downlink precoding. Due to the large number of antennas in massive MIMO the instantaneous channel knowledge, which is commonly assumed to be known perfectly in the power control literature, is hard to be obtained perfectly.  Therefore one needs to take into account both the pilot power and payload power, and hence optimal power control becomes even harder in massive MIMO. This brings a new challenge to designing algorithms for optimal power control to achieve different objectives. On the other hand, the channel hardening in massive MIMO makes it possible to do power control based on the large-scale fading rather than small-scale fading. Several work tried to tackle this hard problem. In \cite{YM2014} the authors optimize the data power for providing every user the same throughput in multi-cell massive MIMO. In \cite{GFGFA2014} power control is done to minimize the uplink power consumption under target SINR constraints.
However we are not aware of any work that jointly optimize the pilot and data power for massive MIMO.
The questions we want to answer in this paper are:
\begin{enumerate}
\item Is power control on the pilots needed for massive MIMO systems? If the answer is yes, how much can we gain from jointly optimizing the pilot power and payload power, as compared to always using full power?
\item What intuition can be obtained from the optimal power control? This includes the pilot length, and how the pilot and payload power depend on the channel quality.
\end{enumerate}
In this paper we provide partial answers to these questions in the case of single cell operation with maximum ratio combining (MRC) at the BS. This is done by formulating and solving the optimization problems and comparing the results with simple heuristic power control policies. Two commonly used performance objectives, namely max-min SE and sum SE optimization,  are investigated and efficient optimization algorithms are developed.

\section{System Model}\label{model}
We consider uplink single cell massive MIMO systems with $M$ antennas at the BS and $K$ single-antenna users. The $K$ users are assigned $K$ orthogonal pilot sequences of length $\tau_p$ for $K \leq \tau_p \leq T$, where $T$ is the length of the coherence interval in which the channels are assumed to be unchanged. Denote the Rayleigh fading channels between the BS and the users as $\bG\in \mathbb{C}^{M \times K}$, where the columns of $\bG$ have the distributions
\begin{equation}
\bg_{k}\sim CN(\mathbf{0},\beta_k\bI), ~k=1,2,\ldots, K,
\end{equation}
which is a circularly symmetric complex Gaussian random variable, and the variance $\beta_k$ represents the large-scale fading including path loss and shadowing. The large-scale fading coefficients are assumed to be known at the BS as they are varying slowly (in the scale of thousands of coherence intervals) and can be easily estimated. The schemes proposed in this paper only depends on the large-scale fading which make it feasible to optimize the power control online.

In each coherence interval, each user $k$ transmits an orthogonal pilot sequence with power $p_p^k$ for channel estimation. We assume that minimum mean squared error (MMSE) channel estimation is carried out at the BS to obtain the small-scale coefficients. This gives an MMSE estimate of the channel vector from user $k$ as
\begin{equation}
\hat{\bg}_k=\frac{\sqrt{\tau_p p_p^k \beta_k}}{1+\tau_p p_p^k \beta_k}\left(\sqrt{\tau_p p_p^k}\bg_k+\bn_p^k\right)
\end{equation}
where $\bn_p^k\sim CN(\bf{0},\bI)$ accounts for the additive noise during the training interval. The noise has been normalized to unit variance and the variance is absorbed into $p_p^k$. During the payload transmission interval, the BS received the signal:
\begin{equation}
\by=\sum_{k=1}^K \bg_k p_u^ks_k+\bn
\end{equation}
where $s_k$ is the unit variance information symbol from user $k$ and $\bn\sim CN(\bf{0},\bI)$ represents the noise during the payload transmission interval. The noise has been normalized to have unit variance and the variance is absorbed into the payload power $p_u^k$. The channel estimates are used for MRC detection of the payload, which corresponds to multiplying the $\by$ with $\hat{\bg}_k^H$ to detect the symbol $s_k$. To make fair comparison with the scheme with equal power allocations, we impose the following constraint on the total transmit energy over a coherence interval:
\begin{equation}\label{power}
\tau_p p_{p}^k+(T-\tau_p) p_u^k \leq E_{max},~k=1,\ldots,K
\end{equation}
where $E_{max}$ is the total energy budget for each user within one coherence interval. In previous work, $p_{p}^k$ and $p_{u}^k$ have been optimized separately or often not optimized at all in which the benefit of massive MIMO cannot be fully harvested. Therefore we consider the scenario where each user can choose freely how to allocate its energy budget on the pilots and payload. In \cite{NML2014} $p_p^k$ and $p_u^k$ are set equal for every user. The work \cite{YM2014} optimized the payload power to maximize the minimum throughput, which corresponds to fixing $p_p^k$ for every user and optimizing over the $p_u^k$. The work \cite{BLD2015} adopted inverse power control for the pilot power, which corresponds to setting $p_p^k=C/\beta_k$ with a normalization constant $C$. These previous work can all be included in our framework by setting different variables to be constant. Therefore our framework of power control is the most general and the algorithms we develop in this work can be applied to all the above scenarios.
\section{Achievable SE With MRC}\label{rate}
Since the exact ergodic capacity of the uplink multiuser channels with channel uncertainty is unknown, bounds on the achievable SE are often adopted as the performance metric in the massive MIMO literature. Here we develop a lower bound for arbitrary power control with the same methodology using the Jensen's inequality as in \cite{NLM2013}. This achievable SE for user $k$ is given by the following proposition.
\begin{proposition}
An achievable SE for user $k$ with power control on the pilot and payload power is given by
\begin{equation}
R_k=\left(1-\frac{\tau_p}{T}\right)\log_2(1+\text{SINR}_k)
\end{equation}
where the SINR of user $k$ is
\begin{equation}\label{sinr}
\text{SINR}_k=\frac{(M-1)p_u^k p_p^k\beta_k^2 \tau_p}{1+\sum_{j=1}^K \beta_j p_u^j +\tau_p \beta_k p_p^k+ \tau_p p_p^k \beta_k \sum_{j\neq k} \beta_j p_u^j}.
\end{equation}
\end{proposition}
\begin{proof}
Using a similar approach as in \cite{NLM2013} by treating all additive interference as a worst-case Gaussian noise and then applying Jensen's inequality we can arrive at the result. The detailed proof is omitted here due to lack of space.
\end{proof}
This achievable SE is used as the user performance metric throughout the paper, where $\tau_p$, $p_p^k$ and $p_u^k$ are the variables to be optimized (for $k=1,\ldots,K$). The optimization can be done at the BS, which can then inform the users about the number of pilot symbols, the amount of power to be spent on training, and the amount of power to be spent on payload data. The aim is to maximize a given utility function $U(R_1,\ldots,R_K)$ where $U(\cdot)$ can be any function that is monotonically increasing in every argument. The utility function characterizes the performance and fairness that we provide to the users. Examples of commonly used utility function are the max-min fairness, sum performance, and proportional fairness. The general problem we are trying to solve is:
\begin{equation}\label{opt}
\begin{aligned}
& \underset{\tau_p,\{p_p^k\},\{p_u^k\}}{\text{maximize}}
& & U ~(R_1,\ldots,R_K)   \\
&~\text{subject to}
& & \tau_p p_{p}^k+(T-\tau_p) p_u^k \leq E_{max}, \forall k \\
& & & p_p^k\geq 0, p_u^k\geq0, \forall k, K\leq\tau_p\leq T.
\end{aligned}
\end{equation}
\section{Optimal Training Interval}\label{training}
In this section we derive the optimal length of the training interval in closed form. First we provide the following lemma:
\begin{lemma}\label{lemma1}
For any monotonically increasing utility function, the energy constraint \eqref{power} is satisfied with equality for every user at the optimal solution, i.e.
\begin{equation}
\tau_p p_{p}^k+(T-\tau_p) p_u^k = E_{max}, ~k=1,\ldots,K
\end{equation}
at the optimal point.
\end{lemma}
\begin{proof}
Observe that we use orthogonal pilot sequences for each user and therefore the SINRs in \eqref{sinr} are monotonically increasing in $p_p^k$ for every user $k$. For any power allocation in which some users do not use the full energy budget, they can each increase their pilot power to improve their own SINR until their energy constraint is satisfied with equality, without causing interference to any other users.
\end{proof}

Then we can state the following theorem which gives the optimal length of training interval in closed form.
\begin{theorem}\label{theorem1}
For any monotonically increasing utility function $U(R_1,\ldots,R_k)$, the optimal $\tau_p$ equals to $K$.
\end{theorem}
\begin{proof}

We prove this by applying that the function $g(x)=\log\left(1+\frac{a}{bx+c}\right)$ is a strictly monotonic increasing function in $x$.  The
detailed proof is omitted here due to lack of space.
\end{proof}
Using the result in Theorem \ref{theorem1}, we can reduce the number of variables involved in the optimization and this enable us to find the optimal solutions for certain utility functions in the next section. Also from Theorem \ref{theorem1} we know that the optimal training period $\tau_p$ is equal to the number of users being served, and is the same for every user. Therefore there is no need for assigning pilot sequences of different length for different users.
\section{Joint Power Control of Pilots and Payload}\label{problem}
In this section we focus on solving the power control problem \eqref{opt} for two different utilities, namely max-min fairness and the sum performance, while other utilities are left for future work. These are the two extreme cases: totally fair and ignoring fairness to achieve high total throughput.
\subsection{Maximize the Minimum SE}
In the max-min fairness which aim at serving every user in the cell with equal SE. This is corresponding to choosing $U(R_1,\ldots,R_K)=\min_k R_k$. By using Theorem \ref{theorem1}, \eqref{opt} now becomes the following optimization problem:
\begin{equation}
\begin{aligned}
& \underset{\{p_p^k\},~\{p_u^k\}}{\text{maximize}}
& & \underset{k}{\text{min}} ~R_k   \\
& \text{subject to}
& & \tau_p p_{p}^k+(T-\tau_p) p_u^k \leq E_{max}, \forall k \\
& & & p_p^k\geq 0, p_u^k\geq0, \forall k.
\end{aligned}
\end{equation}
Since $\log(1+x)$ is an increasing function of $x$, we can remove the logarithm in the objective and use the epigraph form:

\begin{equation}\label{maxmin}
\begin{aligned}
& \underset{\{p_p^k\},\{p_u^k\},~\lambda}{\text{maximize}}
& &  \lambda  \\
& ~\text{subject to}
& & (M-1)p_u^k p_p^k \beta_k^2 \tau_p \geq\\
&&& \lambda(1+\sum_{j=1}^K \beta_j p_u^j +\tau_p \beta_k p_p^k+ \\
&&&\tau_p p_p^k \beta_k \sum_{j\neq k} \beta_j p_u^j), \forall k\\
& & & \tau_p p_{p}^k+(T-\tau_p) p_u^k \leq E_{max}, \forall k \\
& & & p_p^k\geq 0, p_u^k\geq0, \forall k.
\end{aligned}
\end{equation}

This problem is non-convex as it is formulated here, however it is a geometric program (GP). Since the objective function is a monomial and the constraints are valid posynomial, this can be solved efficiently with any GP solvers, here we use the MOSEK solver \cite{MOSEK} with CVX \cite{CVX}. Alternatively, \eqref{maxmin} can be turned into a convex optimization problem by the change of variable: $y_i=\log x_i$ for every variable $x_i$ (in this case $p_p^k, p_u^k,\lambda$). Then any function in the form $m(\bx)=c\prod_i x_i^{a_i}$ with $c>0$, which is defined as a monomial, becomes an exponential of an affine function. Sum of monomials, defined as the posynomial, becomes sum of exponentials of affine functions. Finally taking logarithm of every objective and constraint function we obtained a convex problem.
\subsection{Maximize the Sum SE}
In this part, we aim at maximizing the sum SE by choosing $U(R_1,\ldots,R_K)=\sum_{k=1}^K R_k$. By using Theorem \ref{theorem1}, \eqref{opt} now becomes the following optimization problem:
\begin{equation}\label{sumrate}
\begin{aligned}
& \underset{\{p_p^k\},~\{p_u^k\}}{\text{maximize}}
& &  \sum_k\log(1+\text{SINR}_k)   \\
& \text{subject to}
& & \tau_p p_{p}^k+(T-\tau_p) p_u^k \leq E_{max}, \forall k \\
& & & p_p^k\geq 0, p_u^k\geq0, \forall k.
\end{aligned}
\end{equation}
Power control to maximize sum performance is known to be an $\mathrm{NP}$-hard problem, even under perfect channel knowledge. Therefore in this paper we aim at finding a local optimal solution with affordable computational complexity. Here we adopt a successive convex optimization approach to converge to a Karush-Kuhn-Tucker (KKT) point of \eqref{sumrate}\cite{MW1978}.
We first reformulate the problem using the epigraph form as
\begin{equation}\label{sumrate2}
\begin{aligned}
& \underset{\{p_p^k\},\{p_u^k\},~\{\lambda_k\}}{\text{maximize}}
& &  ~\prod_k \lambda_k   \\
& ~~~\text{subject to}
& & \tau_p p_{p}^k+(T-\tau_p) p_u^k \leq E_{max}, \forall k \\
& & & 1+\text{SINR}_k \geq \lambda_k, \forall k \\
& & & p_p^k\geq 0, p_u^k\geq0, \forall k.
\end{aligned}
\end{equation}
In this form it is clear that the only non-convexity lies in the constraints of the SINRs. To deal with these SINR constraints, we construct a family of functions $f_i(\mathbf{p})$ in each iteration $i$ to approximate $f(\bp_k)=1+\text{SINR}_k$ where we denote $\bp_k=(p_p^k,p_u^k,\lambda_k)$. This has to be done for every user $k$ and the functions need to satisfy the conditions given in the following lemma from \cite{MW1978}:
\begin{lemma}\label{leamm}
By constructing a family of functions satisfying the following conditions:
\begin{enumerate}
\item $f(\mathbf{p}_k)\leq f_i(\mathbf{p}_k)$, $\forall \mathbf{p}_k$ in the feasible set,
\item $f(\mathbf{p}_k^{(i-1)})=f_i(\mathbf{p}_k^{(i-1)})$, where $\mathbf{p}_k^{(i-1)}$ is the solution from the previous iteration,
\item $\nabla f(\mathbf{p}_k^{(i-1)})=\nabla f_i(\mathbf{p}_k^{(i-1)})$,
\end{enumerate}
and optimize the problem by replacing $f(\mathbf{p}_k)$ with $f_i(\mathbf{p}_k)$ in the $i$-th iteration, the series of the solution will converge to an KKT point of the original problem.
\end{lemma}
The first condition is to ensure that the solution we get is feasible for the original problem. The second condition ensures that the solution from the previous iteration is feasible for the current iteration. As a result the objective value of the original problem increases in every iteration since the solution from the previous iteration is a feasible point to the problem in the current iteration. The second and third conditions together guarantee that the KKT conditions for the original problem are satisfied at convergence. As the objective value is bounded from above and monotonically increasing in every iteration, convergence is guaranteed.

To construct the family of function $f(\bx)$ we need the following lemma from \cite{CTPNJ2007}:
\begin{lemma}\label{lemma4}
For any posynomial $g(\bx)=\sum_i m_i(\bx)$, it holds for any $\alpha_i$ that
\begin{equation}
g(\bx)\geq \tilde{g} (\bx)=\prod_i\left(\frac{m_i(\bx)}{\alpha_i}\right)^{\alpha_i}
\end{equation}
\end{lemma}
The SINR constraints in \eqref{sumrate2} are in the form $h(\bx)/g(\bx)\leq 1$, which is not a valid posynomial constraint. We apply Lemma \ref{lemma4} on the denominator to replace $g(\bx)$ with $\tilde{g} (\bx)$ make it a valid posynomial constraint. Moreover with $\alpha_i$ chosen as $\alpha_i=m_i(\bx_0)/g(\bx_0)$, the three conditions in Lemma \ref{leamm} are satisfied. Doing this for every SINR constraint we get a convex approximation of problem \eqref{sumrate2}. The same procedures are repeated until convergence. To conclude, we obtain a KKT point to \eqref{sumrate2} with the procedure described in Algorithm \ref{SCO}.
\alglanguage{pascal}
\begin{algorithm}
 \begin{algorithmic}[1]
 \State choose $\mathbf{p}_k^{(0)}$ as the solution of max-min problem satisfying the constraints and initialize $i=1$
 \Repeat
 \State form the $i$-th approximated problem of \eqref{sumrate2} by approximating every SINR constraints using lemma \ref{lemma4},
 \State solve the $i$-th approximated problem to obtain $\mathbf{p}_k^{(i)}$ for every user $k$,
 \State $i \gets i+1$
 \Until \rm{convergence}
 \State \textbf{return} all $\mathbf{p}_k^{(i)}$
 \end{algorithmic}
 \caption{Successive convex optimization for problem \eqref{sumrate2}}
 \label{SCO}
\end{algorithm}
\section{Simulation Results}\label{simulation}
In this section we present simulation results to demonstrate the benefits of our algorithms and compare the performance with the case of no power control as well as the case of power control on the payload power only. There are $5$ schemes we are comparing here: 1) the solution to problem \eqref{maxmin} (marked as `max-min' in the figures), 2) Algorithm \ref{SCO} for problem \eqref{sumrate2} (marked as `sum' in the figures), 3) equal power allocation $p_u^k=p_p^k=E_{max}/T$ (marked as `no control' in the figures), 4) optimizing only payload power for problem \eqref{maxmin} by fixing $p_p^k=E_{max}/T$ (marked as `max-min (data only)' in the figures), 5) optimizing payload power only for problem \eqref{sumrate2} using Algorithm \ref{SCO} by fixing $p_p^k=E_{max}/T$ (marked as `sum (data only)' in the figures). We consider a scenario with $M=100$ antennas, $K=10$ users, and the length of the coherence interval is $T=200$ (which for example corresponds to a coherence bandwidth of $200$ kHz and a coherence time of $1$ ms). The users are assumed to be uniformly and randomly distributed in a cell with radius $R=500$ m and no user is closer to the BS than $100$ m. The path-loss model is chosen as $\beta_k=1/r_k^{3.76}$ where $r_k$ is the distance of user $k$ from the BS. The constant $E_{max}=0.1\times R^{3.76} \times T$ to get a signal-to-noise ratio (SNR) of $-10$ dB at the cell edge when using equal power allocation. Hence the users that are closer to the BS have higher SNR. The algorithms are run for $1000$ Monte-Carlo simulations where in each snapshot the users are dropped randomly in the cell so that the large-scale fading $\beta_k$ changes.

First we consider the cumulative distribution (CDF) of the sum SE. In Figure \ref{sumratelowsnr} we plot the CDF of the sum SE for the scenario we described. We observe the optimized power control increases the sum SE significantly. The whole CDF is shifted to the right by almost $8$ bit/s/Hz with the proposed power control on both the pilot and data as compared to equal power allocation. For example at the $0.95$-likely point, power control on the data only increases the sum SE by around $45\%$, power control over both pilot and data contributes to another $20\%$ increases as compared to equal power allocation. Surprisingly even the max-min formulation, which is designed for providing fairness, increases the sum SE by around $35\%$ comparing to equal power allocation.

In Figure \ref{minratelowsnr} we plot the CDF of the minimum SE over different snapshots of user locations. We observe that without any power control in half of the cases the user with the lowest SNR will get less than $0.5$ bit/s/Hz. This is not acceptable if we want to provide decent quality of service to every user being served. With max-min power control for both pilot and data we resolve this problem by guaranteeing every user a SE of more than $2$ bit/s/Hz. Moreover the $0.95$-likely point is increased by 9 times with power control on the data only, and 10 times with power control on both pilot and data with respect to equal power allocation. The performance of the sum SE formulation is rather surprising. By optimizing the pilot and data power for the sum SE also increases the minimum SE by 6 times in the $0.95$-likely point and guarantee each user to have more than $1$ bit/s/Hz.

Finally in Figure \ref{peruserratelowsnr} we plot the CDF of the per user SE over different snapshots of user locations. We observe that without power control the SE of the users varies from almost $0$ bit/s/Hz, which is not acceptable, to $10$ bit/s/Hz which is probably wasted in practice due to the limited modulation size. With sum SE power control the problem is less serious as the SE only ranges from $1$ b/s/Hz to $3.8$ b/s/Hz. With the max-min power control we provide almost the same SINR in every snapshot, since the SE of each user is very concentrated at its median point of $2.2$ bit/s/Hz.

In all the figures we observe that our joint power control over both pilots and data behaves similarly to power control over only the data. Nevertheless, our joint optimization increases the benefits even further. The extra gain can be up to $30\%$ if we look at the max-min formulation. The gain of the sum SE is smaller, but can still be up to $20\%$. This shows if we need to optimize the power allocation, power control over both pilots and data is the right approach.

Some additional insights on how the optimal power allocation is can be obtained from the solutions. For the max-min SE formulation, users that are closer to the BS spend more power on the pilots and users that are further away spend more power on the payload. This is due to the fairness property of the max-min formulation, which reduces the interference caused from the users that are closer to the BS in order to reduce the near-far effect. For the sum SE formulation, both users that are very close and further than certain distance threshold $d_{th}$ spend more power on pilots. This is because we require a more accurate channel estimate from the users that are very far and reduce the interference from the users that are very close to the others. This threshold depends on all $\beta_k$ and $E_{max}$. More in depth characterization of $d_{th}$ is left for future work.

\begin{figure}
\includegraphics[width=0.5\textwidth]{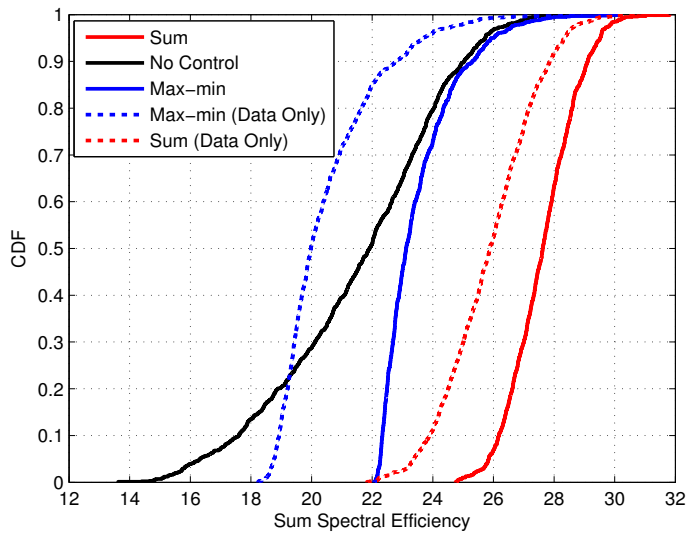}
\caption {\label{sumratelowsnr} CDF of the sum SE with $M=100$, $K=10$, $T=200$, $R=500$m and cell edge SNR of $-10$ dB.}
\end{figure}
\begin{figure}
\includegraphics[width=0.5\textwidth]{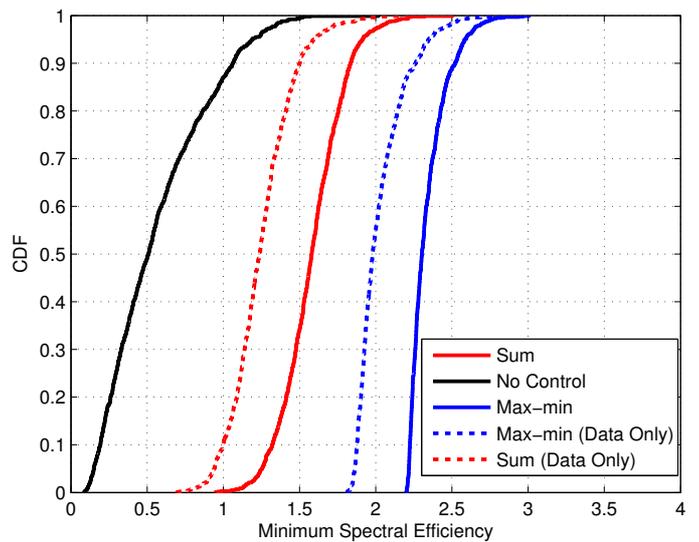}
\caption {\label{minratelowsnr} CDF of the minimum SE with $M=100$, $K=10$ and $T=200$, $R=500$m and cell edge SNR of $-10$ dB.}
\end{figure}
\begin{figure}
\includegraphics[width=0.5\textwidth]{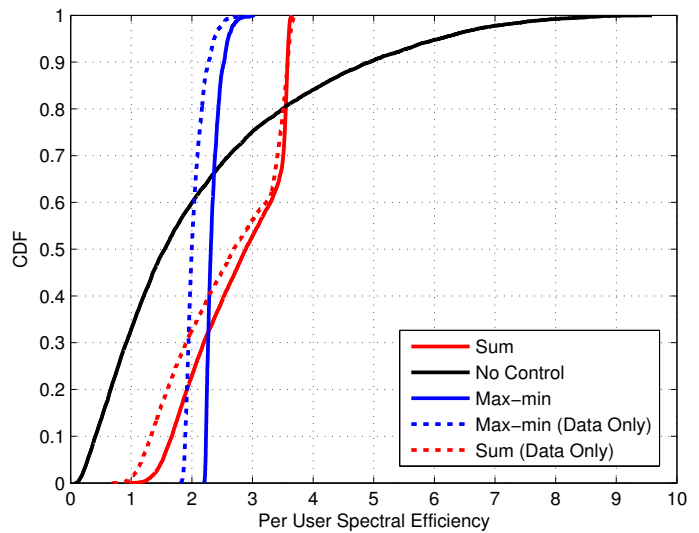}
\caption {\label{peruserratelowsnr} CDF of the per user SE with $M=100$, $K=10$ and $T=200$, $R=500$m and cell edge SNR of $-10$ dB.}
\end{figure}
\section{Conclusion and Future Work}\label{conclusion}
We considered the optimal joint pilot and data power allocation problems in single cell uplink massive MIMO systems with MRC detection. A closed form solution for the optimal length of training interval was first derived. Using the SE as performance metric and setting a total energy budget, the power control was formulated as optimization problems for two different objective functions: the minimum SE and the sum SE. The optimal power control policy was found for the case of maximizing the minimum SE by converting it to a GP. Since maximizing the sum SE is an $\mathrm{NP}$-hard problem, an efficient suboptimal algorithm was developed for finding KKT (local maximum) points. Simulation results showed the advantage of joint optimization over both pilot and data power. The gain over power control on data only can be up to $30\%$ for the minimum SE and $20\%$ for the sum SE.
Future work includes extension to the multi-cell systems and taking into account other detection methods.

\end{document}